\documentclass[%
 aip,
 amsmath,amssymb,
reprint, onecolumn 
]{revtex4-1}

\usepackage{amsthm}
\usepackage{graphicx}
\usepackage{dcolumn}
\usepackage{bm}

\usepackage[utf8]{inputenc}
\usepackage[T1]{fontenc}
\usepackage{mathptmx}
\usepackage{etoolbox}

\usepackage[english]{babel}

\newtheorem{proposition}{Proposition}

\newtheorem{remark}{Remark}

\newtheorem{definition}{Definition}

\makeatletter
\def\@email#1#2{%
 \endgroup
 \patchcmd{\titleblock@produce}
  {\frontmatter@RRAPformat}
  {\frontmatter@RRAPformat{\produce@RRAP{*#1\href{mailto:#2}{#2}}}\frontmatter@RRAPformat}
  {}{}
}%
\makeatother
\begin{document}

\preprint{AIP/123-QED}

\title{Hamiltonian formulation of the quasineutral Vlasov-Poisson system}
\author{J. W. Burby}
\email{joshua.burby@austin.utexas.edu}
\affiliation{ 
Department of Physics and Institute for Fusion Studies, University of Texas at Austin, Austin, TX 78712, USA
}%
\author{D. A. Kaltsas}
\affiliation{Department of Physics, University of Ioannina, Ioannina, GR 451 10, Greece}
\affiliation{Department of Informatics, Democritus University of Thrace, Kavala, GR 654 04, Greece}
\author{P. J. Morrison}
\affiliation{ 
Department of Physics and Institute for Fusion Studies, University of Texas at Austin, Austin, TX 78712, USA
}
\author{E. Tassi}
\affiliation{Universit{\'e} C{\^o}te d'Azure, Observatoire del C{\^o}te d'Azure, CNRS, Laboratoire Lagrange, Bd de l'Observatoire, CS 34229, 06304 Nice Cedex 4, France}
\author{G. N. Throumoulopoulos}
\affiliation{Department of Physics, University of Ioannina, Ioannina, GR 451 10, Greece}


\date{\today}

\begin{abstract}
Slow manifold reduction and the theory of Poisson-Dirac submanifolds are used to deduce a Hamiltonian formulation for a quasineutral limit of the planar, collisionless, magnetized Vlasov-Poisson system. Motion on the slow manifold models plasma dynamics free of fast Langmuir oscillations. Preservation of quasineutrality requires the bulk plasma flow is incompressible. The electric field is determined by counterbalancing plasma stresses that would otherwise produce compression. The Hamiltonian structure for the quasineutral model synthesizes well-known Poisson brackets for incompressible fluids and collisionless kinetic equations. 
\end{abstract}

\maketitle

\section{Introduction} 

Non-relativistic approximations of plasma kinetic theory fall into two broad categories, (quasi) electrostatic and quasineutral. Electrostatic approximations entail nearly irrotational electric fields, while quasineutral approximations entail near local charge neutrality. Each approximation corresponds to a singular limit of a relativistic parent model, such as the Vlasov-Maxwell system. Determining whether a given nonrelativistic plasma process falls cleanly into one category or the other generally requires careful consideration. Some non-relativistic plasmas display significant space-charge effects, suggesting applicability of electrostatics and inapplicability of quasineutrality. Other non-relativistic plasmas exhibit strong inductive electric fields, which fall under the purview of quasineutraility while strongly violating electrostatics. Curiously, the quasineutral and electrostatic approximations can also be applied simultaneously\cite{y_brenier_formulation_1989,e_grenier_oscillations_1996,m_griffin-pickering_recent_2020}. The formulations of gyrokinetic theory in Refs.\,\onlinecite{t_s_hahm_nonlinear_1988,h_sugama_gyrokinetic_2000,brizard_foundations_2007,miyato_gyrokinetic_2013} provide concrete applications of this compound approximation, where the quasineutral approximation is applied after applying an electrostatic approximation with perturbative corrections\cite{sugama_conservation_2013}. 

This article identifies a Hamiltonian formulation for the quasineutral limit of the magnetized, planar Vlasov-Poisson system for electrons moving through a frozen neutralizing ion background. This limit model is referred to as the kinetic incompressible Euler system in Ref.\,\onlinecite{m_griffin-pickering_recent_2020}; it arises by applying the quasineutral approximation \emph{after} the electrostatic approximation. The model's Hamiltonian formulation has gone unnoticed even though comparable formulations for the Vlasov-Maxwell system\cite{p_j_morrison_maxwell-vlasov_1980,j_e_marsden_hamiltonian_1982,morrisonPoissonBracketsFluids1982}, the Vlasov-Poisson system\cite{l_a_turski_canonical_1976,p_j_morrison_maxwell-vlasov_1980,j_gibbons_collisionless_1981}, and the quasineutral Vlasov system\cite{c_tronci_neutral_2015,j_w_burby_chasing_2015} are each known separately. None of these previous formulations directly provides the Hamiltonian structure for the quasineutral electrostatic limit. The derivation of the Hamiltonian structure presented here proceeds by first formulating the quasineutral electrostatic model as a slow manifold reduction\cite{a_n_gorban_constructive_2004,r_s_mackay_slow_2004,j_w_burby_slow_2020} of the electrostatic model and then studying the extrinsic Poisson geometry\cite{a_weinstein_local_1983,vaisman_lectures_1994} of the limiting slow manifold in the infinite-dimensional Vlasov-Poisson phase space. It turns out that the slow manifold comprises an example of a Poisson-Dirac submanifold\cite{crainic_integrability_2004}, and therefore inherits a non-trivial Poisson bracket from the ambient phase space. The novel Hamiltonian formulation is summarized in Propositions \ref{PD_proposition} and \ref{hamiltonian_formulation_proposition} below.

This Article is the first in a two-paper sequence investigating the Hamiltonian structure of quasineutral limits of electrostatic plasma models. It establishes the basic ideas in a simplified setting of two space dimensions and frozen background ions. The second paper \cite{Kaltsas_2025} extends the analysis presented here, allowing for any number of space dimensions, dynamical ions, and the alternative Vlasov-Amp\'ere electrostatic model, before conducting a numerical study of the quasineutral electrostatic model in one space dimension. The second paper also provides a sharper characterization of the quasineutral Hamiltonian  structure. Where this Article shows the quasineutral constraint, comprising local charge neutrality and current incompressibility, defines a Poisson-Dirac submanifold, the analysis in \onlinecite{Kaltsas_2025} demonstrates the constraint in fact defines a Poisson transversal, independently of space dimension. In addition, where this Article identifies the Poisson bracket on the quasineutral constraint manifold itself, Ref.\,\onlinecite{Kaltsas_2025} finds a Poisson bracket in an open neighborhood of the constraint manifold that renders the quasineutral constraint a Casimir invariant.


\section{The magnetized quasineutral Vlasov-Poisson system\label{sec:QNVP}}
The magnetized planar quasineutral Vlasov-Poisson (QNVP) system is defined as follows.
\begin{definition}\label{def1} Fix $\epsilon > 0$ and a nowhere-vanishing function $B:\mathbb{T}^2\rightarrow\mathbb{R}$, where $\mathbb{T}=\mathbb{R} /2\pi\mathbb{Z}$ denotes the $2\pi$-periodic circle.  The magnetized planar QNVP system is given by
\begin{align}
&\partial_tn_0 = 0\label{qn_cont}\\
&\partial_t\bm{\pi} + \epsilon\,\Pi\bigg(\partial_{\bm{q}}\cdot(n_0^{-1}\bm{\pi}\bm{\pi} + n_0\left\langle \bm{\xi}\bm{\xi}\right\rangle)\bigg) =   \Pi(B\mathbb{J}\bm{\pi})\label{qn_vorticity}\\
&\partial_t\varrho+ \partial_{\bm{q}}\cdot\bigg(\epsilon\bigg[n_0^{-1}\bm{\pi} + \bm{\xi}\bigg]\varrho\bigg) + \partial_{\bm{\xi}}\cdot\bigg(\bigg[\epsilon\,\partial_{\bm{q}}\cdot\langle \bm{\xi}\bm{\xi}\rangle - \epsilon\,\bm{\xi}\cdot\partial_{\bm{q}}(n_0^{-1}\bm{\pi}) + B\mathbb{J}\bm{\xi}\bigg]\,\varrho\bigg)\nonumber\\
& - \partial_{\bm{q}}\cdot\bigg(\epsilon\bigg[\langle \bm{\xi}\rangle\bigg]\varrho\bigg) + \partial_{\bm{\xi}}\cdot\bigg(\bigg[\epsilon\,(\partial_{\bm{q}}\langle\bm{\xi}\rangle)\cdot\bm{\xi} - (B-\epsilon\,\Omega)\,\mathbb{J}\langle\bm{\xi}\rangle\bigg]\,\varrho\bigg) = 0,\label{qn_kinetic}
\end{align}
where $n_0\in \mathbb{R}$ represents the electron density, $\bm{\pi}:\mathbb{T}^2\rightarrow\mathbb{R}^2$ is the divergence-free momentum density, and $\varrho:\mathbb{T}^2\times\mathbb{R}^2\rightarrow\mathbb{R}$ denotes the single-electron conditional phase space probability. Points in space are denoted $\bm{q}\in\mathbb{T}^2$, while peculiar velocities are denoted $\bm{\xi}\in \mathbb{R}^2$. The symbol $\Pi$ denotes the $L^2$-projection onto the subspace of divergence-free vector field on $\mathbb{T}^2$. The angle brackets denote conditional expectations, i.e. given $Q:\mathbb{T}^2\times\mathbb{R}^2\rightarrow\mathbb{R}$ the expectation of $Q$ conditioned on $\bm{q}$ is $\langle Q\rangle(\bm{q}) = \int Q(\bm{q},\bm{\xi})\,\varrho(\bm{q},\bm{\xi})\,d\bm{\xi}$. The skew symmetric matrix $\mathbb{J}$ is a counter-clockwise rotation by $\pi/2$ and the fluid vorticity is $\Omega = -\partial_{\bm{q}}\cdot (\mathbb{J}\bm{\pi}/n_0)$.
\end{definition}
\begin{remark}
In this formulation $\varrho$ is an arbitrary positive integrable function on single-particle phase space $\mathbb{T}^2\times\mathbb{R}^2$; it is not required that $\langle 1\rangle=1$ or $\langle \bm{\xi}\rangle = 0$. However it is simple to show that if the latter conditions are satisfied by initial data then they will be satisfied for all time.
\end{remark}
Slow manifold reduction enables deduction of the QNVP system from its parent model, the Vlasov-Poisson (VP) system. This perspective motivates the derivation of QNVP's Hamiltonian formulation in Section \ref{sec:ham_formulation}. Accordingly, this Section presents a derivation of QNVP from VP using the slow manifold idea. It also identifies a presentation of the VP Hamiltonian structure that is well-suited to slow manifold reduction.

In MKS units, the planar, collisionless Vlasov-Poisson system for an electron plasma immersed in a static nowhere-vanishing magnetic field $\bm{B} = B(\bm{q})\,e_z$ is given by
\begin{gather*}
\partial_tf + \bm{v}\cdot\partial_{\bm{q}} f  -q_em_e^{-1}( \partial_{\bm{q}}\varphi + B\mathbb{J}\bm{v})\cdot\partial_{\bm{v}}f = 0\\
-\epsilon_0\Delta \varphi = q_e\int f\,d\bm{v} + \rho_0.
\end{gather*}
Here $f = f(\bm{q},\bm{v})$ denotes the single-electron phase space density and $\varphi = \varphi(\bm{q})$ denotes the electrostatic potential. The electron charge and mass are given by $q_e = -e$ and $m_e$, respectively, while the permittivity of free space is $\epsilon_0$. In the Poisson equation, 
\begin{align*}
\rho_0 = -\frac{q_e\int f\,d\bm{q}\,d\bm{v} }{(2\pi)^2L_0^2}
\end{align*}
represents the uniform neutralizing charge density for a background ion population. The spatial domain $Q\ni \bm{q} = (x,y)$ is assumed doubly-periodic, with period lengths $2\pi L_0$ in each of the $x$- and $y$-directions. The symbol $\Delta= \partial_x^2 + \partial_y^2$ denotes the usual Laplace operator. With the chosen boundary conditions (doubly-periodic), the solution $\varphi$ of Poisson's equation is unique modulo constants. In all that follows the constant is fixed by requiring that the mean electrostatic potential vanishes, $\int\varphi d\bm{q} = 0$.

A dimensionless scaling of the VP system may be formulated as follows. Let $v_0$, $T_0$, $n_0$, $B_0$, and $\varphi_0$ denote characteristic velocity, time, electron density, magnetic field, and electrostatic potential, respectively. Introduce dimensionless field variables using the substitutions
\begin{align*}
f\rightarrow \frac{n_0}{v_0^2}f,\quad B\rightarrow B_0\,B,\quad \varphi\rightarrow \varphi_0\,\varphi,
\end{align*}
and dimensionless independent variables using
\begin{align*}
\bm{q}\rightarrow L_0\,\bm{q},\quad \bm{v}\rightarrow v_0\,\bm{v},\quad t\rightarrow T_0\,t.
\end{align*}
Finally, align the observation timescale with the electron cyclotron period, $T_0 = m_e/(eB_0)$. The VP system may then be written in dimensionless form as
\begin{gather}
\partial_t f + \epsilon\,\bm{v}\cdot\partial_{\bm{q}}f + \bigg(\epsilon\,\Lambda\,\partial_{\bm{q}}\varphi + B\,\mathbb{J}\bm{v} \bigg)\cdot\partial_{\bm{v}}f = 0\label{VP_vlasov}\\
\delta^2\,\Lambda\,\Delta\varphi = \int f\,d\bm{v} - \frac{1}{(2\pi)^2}\int f\,d\bm{q}\,d\bm{v}\label{VP_poisson},
\end{gather}
where three dimensionless parameters appear:
\begin{gather}
\epsilon = \frac{m_ev_0}{e\,B_0\,L_0},\quad \Lambda = \frac{e\varphi_0}{m_ev_0^2},\quad \delta^2 = \frac{\epsilon_0\,m_e\,v_0^2}{e^2\,n_0\,L_0^2}.
\end{gather}
The first two parameters, $\epsilon$ and $\Lambda$, denote ratios of electron gyroradius to field scale length and electron kinetic energy to electrostatic potential energy. Going forward, these parameters will be fixed. The third parameter, $\delta^2$, denotes the (squared) ratio of Debye length to field scale length. It plays a central role in quasineutral asymptotics.

\begin{remark}In the dimensionless VP system, \eqref{VP_vlasov}-\eqref{VP_poisson}, the dimensionless fields $f,\varphi,B$ are doubly-periodic in $\bm{q}$, now with period $2\pi$ in each of the two spatial directions. The potential $\varphi$ is still required to satisfy the normalization condition $\int \varphi d\bm{q} = 0$ in order to ensure uniqueness when solving the Poisson equation.
\end{remark}

The system \eqref{VP_vlasov}-\eqref{VP_poisson} comprises an infinite-dimensional Hamiltonian system on the space $\mathcal{P}_{\text{VP}}$ of distribution functions $f:\mathbb{T}^2\times\mathbb{R}^2\rightarrow \mathbb{R}$, where $\mathbb{T}^2$ denotes the $2$-torus with period $2\pi$ in each direction. The Hamiltonian functional is
\begin{align}
\mathcal{H}_{\text{VP}}(f) = \frac{1}{2}\int |\bm{v}|^2\,f\,d\bm{q}\,d\bm{v} + \delta^2\,\Lambda^2\frac{1}{2}\int \left|\partial_{\bm{q}}\widehat{\varphi}\left(\int f\,d\bm{v}\right)\right|^2\,d\bm{q},\label{VP_hamiltonian}
\end{align}
where $\widehat{\varphi}$ denotes the linear operator that assigns to each $n = \int f\,d\bm{v}$ the unique solution of \eqref{VP_poisson} with $\int \widehat{\varphi}(n)\,d\bm{q} = 0$. The Poisson bracket between functionals $F,G:\mathcal{P}_{\text{VP}}\rightarrow\mathbb{R}$ is given by the well-known formula
\begin{align}
\{F,G\}_{\text{VP}} & = \epsilon\int \bigg(\partial_{\bm{q}}\frac{\delta F}{\delta f}\cdot \partial_{\bm{v}}\frac{\delta G}{\delta f} -\partial_{\bm{q}}\frac{\delta G}{\delta f}\cdot \partial_{\bm{v}}\frac{\delta F}{\delta f} \bigg)\,f\,d\bm{q}\,d\bm{v} + \int B\,\partial_{\bm{v}}\frac{\delta F}{\delta f}\cdot\mathbb{J}\partial_{\bm{v}}\frac{\delta G}{\delta f}\,f\,d\bm{q}\,d\bm{v},\nonumber\\
& = \int \left\{\frac{\delta F}{\delta f},\frac{\delta G}{\delta f}\right\}_0f\,d\bm{q}\,d\bm{v},
\end{align}
where
\begin{align*}
\{h,k\}_0 = \epsilon\bigg(\partial_{\bm{q}}h\cdot\partial_{\bm{v}}k - \partial_{\bm{q}}k\cdot\partial_{\bm{v}}h\bigg) + B\partial_{\bm{v}}h\cdot\mathbb{J}\partial_{\bm{v}}k,
\end{align*}
denotes the single-particle Poisson bracket between functions on phase space, $h,k:\mathbb{T}^2\times \mathbb{R}^2\rightarrow\mathbb{R}$.

In order to deduce the quasineutral limit of the Vlasov-Poisson system it is helpful to modify the formulation of VP at the level of its Hamiltonian structure.  The modification will explicitly separate the density and momentum density evolution from higher-order moments using Lie-theoretic methods developed by Krishnaprasad-Marsden\cite{krishnaprasad_hamiltonian_1987} and C. Tronci\cite{tronci_hamiltonian_2010}. To that end, consider the pair of infinite-dimensional Lie algebras
\begin{align*}
\mathfrak{h} = \{(\psi,\bm{u})\mid \psi:\mathbb{T}^2\rightarrow \mathbb{R},\quad \bm{u}:\mathbb{T}^2\rightarrow\mathbb{R}^2\},\quad \mathfrak{g} = \{\chi\mid \chi:\mathbb{T}^2\times\mathbb{R}^2\rightarrow\mathbb{R}\}.
\end{align*}
The Lie brackets on $\mathfrak{h}$ and $\mathfrak{g}$ are given by
\begin{align*}
[(\psi_1,\bm{u}_1),(\psi_2,\bm{u}_2)]_{\mathfrak{h}} &= -\bigg(\epsilon\left(\bm{u}_1\cdot\partial_{\bm{q}}\psi_2 - \bm{u}_2\cdot\partial_{\bm{q}}\psi_1\right) - B\bm{u}_1\cdot\mathbb{J}\bm{u}_2,\quad\epsilon\left(\bm{u}_1\cdot\partial_{\bm{q}}\bm{u}_2 - \bm{u}_2\cdot\partial_{\bm{q}}\bm{u}_1\right)\bigg)\\
[\chi_1,\chi_2]_{\mathfrak{g}} &=\{\chi_1,\chi_2\}_0.
\end{align*}
The algebra $\mathfrak{h}$ arises as a subalgebra of $\mathfrak{g}$ by identifying each pair $(\psi,\bm{u})$ with the phase space function $\psi(\bm{q})+ \bm{v}\cdot\bm{u}(\bm{q})$. Each element of $\mathfrak{h}$ therefore acts as a (Lie) derivation on $\mathfrak{g}$ according to
\begin{align*}
(\psi,\bm{u})\cdot\chi = \{\psi(\bm{q}) + \bm{v}\cdot\bm{u}(\bm{q}),\chi\}_0.
\end{align*}
It follows that the space $\mathfrak{s} = \mathfrak{h}\times\mathfrak{g}$ enjoys a semi-direct product Lie algebra structure, with Lie bracket given by
\begin{align*}
[((\psi_1,\bm{u}_1),\chi_1),((\psi_2,\bm{u}_2),\chi_2)]_{\mathfrak{s}} = ([(\psi_1,\bm{u}_1),(\psi_2,\bm{u}_2)]_{\mathfrak{h}},[\chi_1,\chi_2]_{\mathfrak{g}} + (\psi_1,\bm{u}_1)\cdot \chi_2 - (\psi_2,\bm{u}_2)\cdot \chi_1).
\end{align*}
The dual space $\mathfrak{s}^* = \mathfrak{h}^*\times\mathfrak{g}^*$ may therefore be equipped with a corresponding Lie-Poisson bracket. Elements of $\mathfrak{s}^*$ are triples $(n,\bm{P},f)$, where $n$ is a function on $\mathbb{T}^2$, $\bm{P}$ is a vector field on $\mathbb{T}^2$, and $f$ is a function on single-particle phase space $\mathbb{T}^2\times\mathbb{R}^2$. The duality pairing between $\mathfrak{s}^*$ and $\mathfrak{s}$ is given explicitly by
\begin{align*}
\left\langle (n,\bm{P},f),(\psi,\bm{u},\chi) \right\rangle  = \int \psi\,n\,d\bm{q} + \int \bm{u}\cdot\bm{P}\,d\bm{q} + \int \chi\,f\,d\bm{q}\,d\bm{v}.
\end{align*}
The Lie-Poisson bracket between functions on $\mathfrak{s}^*$, $F,G:\mathfrak{s}^*\rightarrow\mathbb{R}$ is
\begin{align}
\{F,G\}_{\mathfrak{s}^*} &= -\int \bigg(\epsilon\left(\frac{\delta F}{\delta\bm{P}}\cdot\partial_{\bm{q}}\frac{\delta G}{\delta n} - \frac{\delta G}{\delta\bm{P}}\cdot\partial_{\bm{q}}\frac{\delta F}{\delta n}\right) - B\,\frac{\delta F}{\delta\bm{P}}\cdot \mathbb{J}\frac{\delta G}{\delta\bm{P}}\bigg)\,n\,d\bm{q}\nonumber\\
& -\int \epsilon\bigg(\frac{\delta F}{\delta\bm{P}}\cdot\partial_{\bm{q}}\frac{\delta G}{\delta\bm{P}}-\frac{\delta G}{\delta\bm{P}}\cdot\partial_{\bm{q}}\frac{\delta F}{\delta\bm{P}}\bigg)\cdot\bm{P}\,d\bm{q}\nonumber\\
& + \int \bigg(\left\{\frac{\delta F}{\delta f},\frac{\delta G}{\delta f}\right\}_0 + \left\{\frac{\delta F}{\delta n} + \bm{v}\cdot\frac{\delta F}{\delta \bm{P}},\frac{\delta G}{\delta f}\right\}_0-\left\{\frac{\delta G}{\delta n} + \bm{v}\cdot\frac{\delta G}{\delta \bm{P}},\frac{\delta F}{\delta f}\right\}_0\bigg)\,f\,d\bm{q}\,d\bm{v}.\label{s_bracket}
\end{align}
This algebraic structure is helpful because the map $\mathcal{C}:\mathcal{P}_{\text{VP}}\rightarrow\mathfrak{s}^*$ given by
\begin{align*}
f\mapsto \bigg(\int f\,d\bm{v},\int \bm{v}\,f\,d\bm{v},f\bigg)
\end{align*}
is Poisson due to the commutation relation 
\begin{align*}
\{\langle \mathcal{C},(\psi_1,\bm{u}_1,\chi_1)\rangle,\langle \mathcal{C},(\psi_2,\bm{u}_2,\chi_2)\rangle \}_{\mathcal{P}_{\text{VP}}} = \langle \mathcal{C},[(\psi_1,\bm{u}_1,\chi_1),(\psi_2,\bm{u}_2,\chi_2)]_{\mathfrak{s}}\rangle.
\end{align*}
The Poisson property allows for replacing Hamilton's equations on $\mathcal{P}_{\text{VP}}$ with Hamilton's equations on $\mathfrak{s}^*$ provided the VP Hamiltonian $\mathcal{H}_{\text{VP}}$ relates to a Hamiltonian on $\mathcal{H}_{\mathfrak{s}^*}$ on $\mathfrak{s}^*$ by the Guillemin-Sternberg collectivization\cite{guillemin_moment_1980} formula $\mathcal{H}_{\mathcal{P}_{\text{VP}}} = \mathcal{H}_{\mathfrak{s}^*}\circ \mathcal{C}$. But this can be achieved using
\begin{align*}
\mathcal{H}_{\mathfrak{s}^*}(n,\bm{P},f) = \frac{1}{2}\int |\bm{v} - n^{-1}\bm{P}|^2\,f\,d\bm{v}\,d\bm{q} + \frac{1}{2}\int|n^{-1}\bm{P}|^2\,n\,d\bm{q} + \delta^2\,\Lambda^2\frac{1}{2}\int |\partial_{\bm{q}}\widehat{\varphi}(n)|^2\,d\bm{q}.
\end{align*}
It follows that any solution of Hamilton's equations on $\mathcal{P}_{\text{VP}}$ with Hamiltonian $\mathcal{H}_{\mathcal{P}_{\text{VP}}}$ maps (along $\mathcal{C}$) to a solution of Hamilton's equations with Hamiltonian $\mathcal{H}_{\mathfrak{s}^*}$ on $\mathfrak{s}^*$. This justifies parting ways with Hamilton's equations on $\mathcal{P}_{\text{VP}}$, i.e. Eqs.\,\eqref{VP_vlasov}-\eqref{VP_poisson}, in favor of Hamilton's equations on $\mathfrak{s}^*$, with Poisson bracket \eqref{s_bracket} and Hamiltonian $\mathcal{H}_{\mathfrak{s}^*}(n,\bm{P},f)$. 
\begin{remark}
Note that the space $\mathfrak{s}^*$ is larger than $\mathcal{P}_{\text{VP}}$ because, for an arbitrary $(n,\bm{P},f)\in\mathfrak{s}^*$, it need not be true that $n = \int f\,d\bm{v}$ and $\bm{P} = \int \bm{v}f\,d\bm{v}$. In other words, each state in $\mathfrak{s}^*$ contains redundant information. This redundancy can be eliminated by choosing initial conditions for Hamilton's equations on $\mathfrak{s}^*$ in the image of the Poisson map $\mathcal{C}$.
\end{remark}

The reformulation of Vlasov-Poisson on $\mathfrak{s}^*\ni (n,\bm{P},f)$ just obtained does not \emph{cleanly} separate the first two moments, $(n,\bm{P})$, from all higher-order moments since $f$ (redundantly) encodes both $n$ and $\bm{P}$. A different parameterization of $\mathfrak{s}^*$ alleviates this tension. Consider the invertible transformation $E:\mathfrak{s}^*\rightarrow\mathfrak{s}^*$ defined according to $E(n,\bm{P},f) = (n,\bm{P},\widehat{\varrho}(n,\bm{P},f))$, with
\begin{align*}
\widehat{\varrho}(n,\bm{P},f)(\bm{q},\bm{\xi}) = f(\bm{q},\bm{\xi} + \bm{P}(\bm{q})/n(\bm{q}))/n(\bm{q}).
\end{align*}
For $(n,\bm{P},f)$ in the image of $\mathcal{C}$, the field $\varrho(\bm{q}, \bm{\xi})$, $\bm{\xi}\in\mathbb{R}^2$, represents the centered electron velocity distribution, conditioned on $\bm{q}$. Under this invertible transformation from $(n,\bm{P},f)$-space to $(n,\bm{P},\varrho)$-space, the Hamiltonian function transforms from $\mathcal{H}_{\mathfrak{s}^*}$ to $\mathcal{H}_E$, where
\begin{align}
\mathcal{H}_E(n,\bm{P},\varrho) &=\mathcal{H}_{\mathfrak{s}^*}(n,\bm{P},n\,\tau_{n^{-1}\bm{P}}^*\varrho)\nonumber \\
& = \frac{1}{2}\int |\bm{\xi}|^2\,n\,\varrho\,d\bm{\xi}\,d\bm{q} + \frac{1}{2}\int|n^{-1}\bm{P}|^2\,n\,d\bm{q} + \delta^2\,\Lambda^2\frac{1}{2}\int |\partial_{\bm{q}}\widehat{\varphi}(n)|^2\,d\bm{q}.\label{VP_chamiltonian}
\end{align}
Here $\tau_{n^{-1}\bm{P}}(\bm{q},\bm{v}) = (\bm{q},\bm{v} - n(\bm{q})^{-1}\bm{P}(\bm{q}))$ and $n\,\tau_{n^{-1}\bm{P}}^*\varrho(\bm{q},\bm{v}) = n(\bm{q})\varrho(\bm{q},\bm{v} - n^{-1}\bm{P})$. The Poisson bracket transforms from $\{\cdot,\cdot\}_{\mathfrak{s}^*}$ to $\{\cdot,\cdot\}_E$, where $\{F,G\}_E = E_*\{E^*F,E^*G\}_{\mathfrak{s}^*}$.
By the chain rule, the functional derivatives of $\overline{F} = E^*F$ are given explicitly by
\begin{gather*}
\frac{\delta \overline{F}}{\delta \varrho} = \tau_{n^{-1}\bm{P}}^*\left(\frac{1}{n}\frac{\delta F}{\delta\varrho}\right),\quad \frac{\delta\overline{F}}{\delta \bm{P}} = \frac{\delta F}{\delta\bm{P}} - \left\langle\partial_{\bm{\xi}}\frac{1}{n}\frac{\delta F}{\delta\varrho}\right\rangle\\
\frac{\delta \overline{F}}{\delta n} = \frac{\delta F}{\delta n} - \left\langle \frac{1}{n}\frac{\delta F}{\delta\varrho}\right\rangle + n^{-1}\bm{P}\cdot\left\langle\partial_{\bm{\xi}}\frac{1}{n}\frac{\delta F}{\delta\varrho}\right\rangle,\quad \langle Q\rangle = \int Q\,\varrho\,d\bm{\xi}.
\end{gather*}
The Poisson bracket $\{\cdot,\cdot\}_E$ is therefore given explicitly by
\begin{align}
&\{F,G\}_{E} = -\epsilon\int \bigg(\bigg[\frac{\delta F}{\delta\bm{P}} - \bigg\langle\partial_{\bm{\xi}}\frac{1}{n}\frac{\delta F}{\delta\varrho}\bigg\rangle\bigg]\cdot\partial_{\bm{q}}\bigg[\frac{\delta G}{\delta n} - \left\langle \frac{1}{n}\frac{\delta G}{\delta\varrho}\right\rangle + n^{-1}\bm{P}\cdot\left\langle\partial_{\bm{\xi}}\frac{1}{n}\frac{\delta G}{\delta\varrho}\right\rangle\bigg]\nonumber\\
&\quad\quad\quad\qquad\quad\hspace*{.25em}\, - \bigg[\frac{\delta G}{\delta\bm{P}} - \left\langle\partial_{\bm{\xi}}\frac{1}{n}\frac{\delta G}{\delta\varrho}\right\rangle\bigg]\cdot\partial_{\bm{q}}\left[\frac{\delta F}{\delta n} - \left\langle \frac{1}{n}\frac{\delta F}{\delta\varrho}\right\rangle + n^{-1}\bm{P}\cdot\left\langle\partial_{\bm{\xi}}\frac{1}{n}\frac{\delta F}{\delta\varrho}\right\rangle\right] \bigg) \,n\,d\bm{q}\nonumber\\
&+\int B\,\bigg[\frac{\delta F}{\delta\bm{P}} - \left\langle\partial_{\bm{\xi}}\frac{1}{n}\frac{\delta F}{\delta\varrho}\right\rangle\bigg]\cdot \mathbb{J}\bigg[\frac{\delta G}{\delta\bm{P}} - \left\langle\partial_{\bm{\xi}}\frac{1}{n}\frac{\delta G}{\delta\varrho}\right\rangle\bigg]\,n\,d\bm{q}\nonumber\\
& -\epsilon\int\left[\frac{\delta F}{\delta\bm{P}} - \left\langle\partial_{\bm{\xi}}\frac{1}{n}\frac{\delta F}{\delta\varrho}\right\rangle,\frac{\delta G}{\delta\bm{P}} - \left\langle\partial_{\bm{\xi}}\frac{1}{n}\frac{\delta G}{\delta\varrho}\right\rangle\right]\cdot\bm{P}\,d\bm{q}\nonumber\\
& + \int \left\{\frac{1}{n}\frac{\delta F}{\delta \varrho},\frac{1}{n}\frac{\delta G}{\delta \varrho}\right\}_e \,n\,\varrho\,d\bm{q}\,d\bm{\xi}\nonumber\\
&+\int\bigg(\left\{\frac{\delta F}{\delta n} - \left\langle \frac{1}{n}\frac{\delta F}{\delta\varrho}\right\rangle + n^{-1}\bm{P}\cdot\frac{\delta F}{\delta \bm{P}} + \bm{\xi}\cdot\frac{\delta F}{\delta\bm{P}} - \bm{\xi}\cdot\left\langle\partial_{\bm{\xi}}\frac{1}{n}\frac{\delta F}{\delta\varrho}\right\rangle,\frac{1}{n}\frac{\delta G}{\delta \varrho}\right\}_e\nonumber\\
&\qquad-\left\{\frac{\delta G}{\delta n} - \left\langle \frac{1}{n}\frac{\delta G}{\delta\varrho}\right\rangle + n^{-1}\bm{P}\cdot\frac{\delta G}{\delta \bm{P}} + \bm{\xi}\cdot\frac{\delta G}{\delta\bm{P}} - \bm{\xi}\cdot\left\langle\partial_{\bm{\xi}}\frac{1}{n}\frac{\delta G}{\delta\varrho}\right\rangle,\frac{1}{n}\frac{\delta F}{\delta \varrho}\right\}_e\bigg)n\,\varrho\,d\bm{q}\,d\bm{\xi}.
\end{align}
Here, the single-electron bracket $\{\cdot,\cdot\}_e$ is given by
\begin{align*}
\{h,k\}_e & = \epsilon\bigg(\partial_{\bm{q}}h\cdot\partial_{\bm{\xi}}k - \partial_{\bm{q}}k\cdot\partial_{\bm{\xi}}h\bigg) + (B - \epsilon\,\Omega)\partial_{\bm{\xi}}h\cdot\mathbb{J}\partial_{\bm{\xi}}k,\quad \Omega = -\partial_{\bm{q}}\cdot(\mathbb{J}n^{-1}\bm{P}),
\end{align*}
where the scalar $\Omega = \text{curl}(\bm{P}/n)$ denotes the fluid vorticity. This expression for the Poisson bracket $\{\cdot,\cdot\}_E$ does not explicitly account for certain cancellations that simplify subsequent calculations and clarify the relationship between $\{\cdot,\cdot\}_E$ and well-known bracket formulas for other fluid and kinetic plasma models. Carefully accounting for the cancellations leads to the somewhat simplified formula
\begin{align}
&\{F,G\}_{E} =-\epsilon\int\bigg(\frac{\delta F}{\delta\bm{P}}\cdot \bigg[\partial_{\bm{q}}\frac{\delta G}{\delta n} - \partial_{\bm{q}}\left\langle\frac{1}{n}\frac{\delta G}{\delta\varrho}\right\rangle + \left\langle\partial_{\bm{q}}\frac{1}{n}\frac{\delta G}{\delta\varrho}\right\rangle\bigg]-\frac{\delta G}{\delta\bm{P}}\cdot \bigg[\partial_{\bm{q}}\frac{\delta F}{\delta n} - \partial_{\bm{q}}\left\langle\frac{1}{n}\frac{\delta F}{\delta\varrho}\right\rangle + \left\langle\partial_{\bm{q}}\frac{1}{n}\frac{\delta F}{\delta\varrho}\right\rangle\bigg]\bigg)\,n\,d\bm{q}\nonumber\\
&-\epsilon\int\bigg[\frac{\delta F}{\delta\bm{P}},\frac{\delta G}{\delta\bm{P}}\bigg]\cdot\bm{P}\,d\bm{q} +\int B\,\frac{\delta F}{\delta\bm{P}}\cdot \mathbb{J}\frac{\delta G}{\delta\bm{P}} \,n\,d\bm{q}\nonumber\\
&+\epsilon\int \bigg(\partial_{\bm{q}}\frac{1}{n}\frac{\delta F}{\delta \varrho}\cdot \partial_{\bm{\xi}}\frac{1}{n}\frac{\delta G}{\delta\varrho} -\partial_{\bm{q}}\frac{1}{n}\frac{\delta G}{\delta \varrho}\cdot \partial_{\bm{\xi}}\frac{1}{n}\frac{\delta F}{\delta\varrho} \bigg)\,n\,\varrho\,d\bm{\xi}\,d\bm{q} + \int (B-\epsilon\,\Omega)\bigg(\partial_{\bm{\xi}}\frac{1}{n}\frac{\delta F}{\delta\varrho}\bigg)\cdot\mathbb{J}\bigg(\partial_{\bm{\xi}}\frac{1}{n}\frac{\delta G}{\delta\varrho}\bigg)\,n\,\varrho\,d\bm{\xi}\,d\bm{q}\nonumber\\
&-\epsilon\int \bigg(\left\langle\partial_{\bm{q}}\frac{1}{n}\frac{\delta F}{\delta\varrho}\right\rangle\cdot\left\langle\partial_{\bm{\xi}}\frac{1}{n}\frac{\delta G}{\delta\varrho}\right\rangle-\left\langle\partial_{\bm{q}}\frac{1}{n}\frac{\delta G}{\delta\varrho}\right\rangle\cdot\left\langle\partial_{\bm{\xi}}\frac{1}{n}\frac{\delta F}{\delta\varrho}\right\rangle\bigg)\,n\,d\bm{q}-\int (B-\epsilon\,\Omega)\, \left\langle\partial_{\bm{\xi}}\frac{1}{n}\frac{\delta F}{\delta\varrho}\right\rangle\cdot \mathbb{J} \left\langle\partial_{\bm{\xi}}\frac{1}{n}\frac{\delta G}{\delta\varrho}\right\rangle\,n\,d\bm{q}\nonumber\\
&- \epsilon \int \bigg(\left[\frac{\delta F}{\delta\bm{P}} - \left\langle\partial_{\bm{\xi}}\frac{1}{n}\frac{\delta F}{\delta\varrho}\right\rangle\right]\cdot\partial_{\bm{q}}\cdot\left[n\left\langle\partial_{\bm{\xi}}\left(\frac{1}{n}\frac{\delta G}{\delta\varrho}\right)\bm{\xi}\right\rangle\right] -\left[\frac{\delta G}{\delta\bm{P}} - \left\langle\partial_{\bm{\xi}}\frac{1}{n}\frac{\delta G}{\delta\varrho}\right\rangle\right]\cdot\partial_{\bm{q}}\cdot\left[n\left\langle\partial_{\bm{\xi}}\left(\frac{1}{n}\frac{\delta F}{\delta\varrho}\right)\bm{\xi}\right\rangle\right] \bigg)\,d\bm{q}.\label{VP_cbracket}
\end{align}

Using these formulas for the Poisson bracket $\{\cdot,\cdot\}_E$ and Hamiltonian $\mathcal{H}_E$ on $(n,\bm{P},\varrho)$-space, Hamilton's equations become
\begin{gather}
\partial_tn + \epsilon\,\partial_{\bm{q}}\cdot\bm{P}  = 0,\quad \partial_t\bm{P} + \epsilon\,\partial_{\bm{q}}\cdot(n^{-1}\bm{P}\bm{P}) = - \epsilon\,\partial_{\bm{q}}\cdot(n\langle \bm{\xi}\bm{\xi}\rangle) + \epsilon\,\Lambda\,n\,\partial_{\bm{q}}\widehat{\varphi}(n) + B\mathbb{J}\bm{P},\label{fs_fluid}\\
\partial_t(n\varrho) + \partial_{\bm{q}}\cdot\bigg(\epsilon\bigg[n^{-1}\bm{P} + \bm{\xi}\bigg]n\varrho\bigg) + \partial_{\bm{\xi}}\cdot\bigg(\bigg[\epsilon\,n^{-1}\partial_{\bm{q}}\cdot(n\langle \bm{\xi}\bm{\xi}\rangle) - \epsilon\,\bm{\xi}\cdot\partial_{\bm{q}}(n^{-1}\bm{P}) + B\mathbb{J}\bm{\xi}\bigg]\,n\varrho\bigg)\nonumber\\
 - \partial_{\bm{q}}\cdot\bigg(\epsilon\bigg[\langle \bm{\xi}\rangle\bigg]n\varrho\bigg) + \partial_{\bm{\xi}}\cdot\bigg(\bigg[\epsilon\,(\partial_{\bm{q}}\langle\bm{\xi}\rangle)\cdot\bm{\xi} - (B-\epsilon\,\Omega)\,\mathbb{J}\langle\bm{\xi}\rangle\bigg]n\,\varrho\bigg) = 0,\label{fs_kinetic}
\end{gather}
which provides the ultimate reformulation of the VP system that will be referred to in the remainder of this Article. For later reference, it is also convenient to record the form of Hamilton's equations with a general Hamiltonian $G:\mathfrak{s}^*\rightarrow\mathbb{R}$. The result is
\begin{align}
&\partial_tn +\partial_{\bm{q}}\cdot\left(n\frac{\delta G}{\delta\bm{P}}\right) = 0 \label{ndot_G}\\
&\partial_t\bm{P} + \epsilon\,n\,\partial_{\bm{q}}\left(\frac{\delta G}{\delta\bm{P}}\cdot \frac{\bm{P}}{n}\right) + \epsilon\,\partial_{\bm{q}}\cdot\left(n\,\frac{\delta G}{\delta\bm{P}}\right)\frac{\bm{P}}{n} + \epsilon\,n\,\Omega\mathbb{J}\frac{\delta G}{\delta \bm{P}} =-\epsilon\,\partial_{\bm{q}}\cdot\left(n\left\langle\partial_{\bm{\xi}}\left(\frac{1}{n}\frac{\delta G}{\delta\varrho}\right)\bm{\xi} \right\rangle\right) \nonumber\\
&\hspace*{21em}- \epsilon\,n\,\bigg(\partial_{\bm{q}}\frac{\delta G}{\delta n} - \partial_{\bm{q}}\left\langle \frac{1}{n}\frac{\delta G}{\delta\varrho}\right\rangle + \left\langle\partial_{\bm{q}}\frac{1}{n}\frac{\delta G}{\delta\varrho} \right\rangle\bigg) + n\,B\,\mathbb{J}\frac{\delta G}{\delta\bm{P}} \label{Pdot_G}\\
&\partial_t(n\varrho) + \partial_{\bm{q}}\cdot\bigg(\bigg[\epsilon\,\frac{\delta G}{\delta\bm{P}} +\epsilon\partial_{\bm{\xi}}\frac{1}{n}\frac{\delta G}{\delta\varrho}- \epsilon\left\langle\partial_{\bm{\xi}}\frac{1}{n}\frac{\delta G}{\delta\varrho}\right\rangle\bigg]\,n\varrho\bigg)\nonumber\\
&+ \partial_{\bm{\xi}}\cdot\bigg(\bigg[-\epsilon\,\partial_{\bm{q}}\left(\bm{\xi}\cdot\frac{\delta G}{\delta\bm{P}}\right)+\epsilon\,\partial_{\bm{q}}\left(\bm{\xi}\cdot \left\langle\partial_{\bm{\xi}}\frac{1}{n}\frac{\delta G}{\delta\varrho}\right\rangle\right) + \epsilon\,n^{-1}\,\partial_{\bm{q}}\cdot\left(n\left\langle\partial_{\bm{\xi}}\left(\frac{1}{n}\frac{\delta G}{\delta\varrho}\right)\bm{\xi}\right\rangle\right)\nonumber\\
&\qquad\qquad-\epsilon\,\partial_{\bm{q}}\left(\frac{1}{n}\frac{\delta G}{\delta\varrho}\right) + \epsilon\,\left\langle\partial_{\bm{q}}\left(\frac{1}{n}\frac{\delta G}{\delta\varrho}\right)\right\rangle + (B - \epsilon\Omega)\mathbb{J}\bigg(\partial_{\bm{\xi}}\frac{1}{n}\frac{\delta G}{\delta \varrho} - \left\langle \partial_{\bm{\xi}}\frac{1}{n}\frac{\delta G}{\delta \varrho}\right\rangle\bigg)\bigg]\,n\varrho\bigg) = 0.\label{varrhodot_G}
\end{align}

Roughly speaking, the QNVP system is the limit of the VP system as $\delta\rightarrow 0$, with $\epsilon$ and $\Lambda$ held fixed. This limit is singular because the operator $\widehat{\varphi}(n) = O(1/\delta^2)$, which implies an exploding electrostatic force unless the electron charge density very nearly cancels the background ion charge density. So it is more precise to say that the QNVP system is the model that describes solutions of the VP system that remain regular as $\delta\rightarrow 0$. 

Suppose $(n_\delta,\bm{P}_\delta,\varrho_\delta)$ is a $\delta$-dependent solution of the (reformulated) VP system \eqref{fs_fluid}-\eqref{fs_kinetic} that is regular as $\delta\rightarrow 0$. In particular assume that $\partial_tn_\delta$, $\partial_t\bm{P}_\delta$, and $\partial_t{\varrho}_\delta$ are each $O(1)$ in the limit. Let $(n,\bm{P},\varrho) = (n_0,\bm{P}_0,\varrho_0)$. Then multiplying the divergence of the momentum equation by $\delta^2$ and sending $\delta\rightarrow 0$ implies $n - (2\pi)^{-2}\int n\,d\bm{q} = 0$; the electron charge density exactly cancels the background ion charge density in the limit. Since $\partial_tn + \epsilon\,\partial_{\bm{q}}\cdot\bm{P} = 0$, the time derivative of the charge neutrality relation implies $\partial_{\bm{q}}\cdot\bm{P} = 0$. Thus, the electron momentum density must be divergence-free in the limit. Equivalently (by constancy of limiting electron density), the limiting electron flow must be incompressible. 

The preceding pair of observations motivates exchanging the dependent variables $(n,\bm{P})$ with a nicer set that ``blows up" the singularity $\delta\rightarrow 0$. Suppose $\bm{P}$ is the momentum density at some time. By the Hodge decomposition on $\mathbb{T}^2$ there is a unique function $\Phi:\mathbb{T}^2\rightarrow\mathbb{R}$ with vanishing mean and a unique divergence-free vector field $\bm{\pi}$ such that $\bm{P} =\partial_{\bm{q}}\Phi + \bm{\pi} $. Similarly, the electron density decomposes uniquely as $n = n_0 + \delta\,\widetilde{n}$, where $n_0$ is spatially-constant and $\widetilde{n}$ has vanishing mean. When expressed in terms of $(\Phi,\bm{\pi},n_0,\widetilde{n})$ the fluid moment evolution equations become
\begin{align}
&\partial_tn_0 = 0\label{fs_mean_density}\\
& \partial_t\widetilde{n} + \frac{\epsilon}{\delta}\Delta\Phi = 0\label{fs_continuity}\\
& \partial_t \Delta\Phi + \epsilon\,\partial_{\bm{q}}\partial_{\bm{q}}:(n^{-1}\bm{P}\bm{P} + n\left\langle\bm{\xi}\bm{\xi}\right\rangle) = \frac{\epsilon}{\delta}n_0\,\widetilde{n} + \epsilon\,\partial_{\bm{q}}\widetilde{n}\cdot \partial_{\bm{q}}\mathcal{G}[\widetilde{n}] + \partial_{\bm{q}}\cdot(B\mathbb{J}\bm{P})\label{fs_potential}\\
&\partial_t\bm{\pi} + \epsilon\,\Pi\bigg(\partial_{\bm{q}}\cdot(n^{-1}\bm{P}\bm{P} + n\left\langle \bm{\xi}\bm{\xi}\right\rangle)\bigg) = \epsilon\,\Pi\bigg(\widetilde{n}\partial_{\bm{q}}\mathcal{G}[\widetilde{n}]\bigg) + \Pi(B\mathbb{J}\bm{P}),\label{fs_vorticity}
\end{align}
where $\mathcal{G}$ denotes the inverse of $\Delta$ regarded as a self-adjoint operator on the space of functions $\mathbb{T}^2\rightarrow\mathbb{R}$ with zero mean, and $\Pi$ denotes the $L^2$-orthogonal projection onto divergence-free vector fields $\mathbb{T}^2\rightarrow\mathbb{R}^2$. In conjunction with Eq.\,\eqref{fs_kinetic}, Eqs.\,\eqref{fs_mean_density}-\eqref{fs_vorticity} comprise a fast-slow system on $(n_0,\widetilde{n},\Phi,\bm{\pi},\varrho)$-space, in the sense described in Ref.\,\onlinecite{j_w_burby_slow_2020}. The slow variable is $x=(n_0,\bm{\pi},\varrho)$ while the fast variable is $y = (\widetilde{n},\Phi)$. Notice that on the $O(\epsilon)$ timescale the slow variable $x$ is frozen while the fast variable $y$ obeys 
\begin{align}
\partial_t\widetilde{n} =  -\frac{\epsilon}{\delta}\Delta\Phi,\quad \partial_t\Delta\Phi = \frac{\epsilon}{\delta}n_0\,\widetilde{n}.\label{nearly-periodic}
\end{align}
These equations describe the Langmuir oscillation, with frequency $ \frac{\epsilon}{\delta}\sqrt{n_0}$. (Recall that time is measured in units of the electron cyclotron frequency.) As for all fast-slow systems, this system admits a formal slow manifold of the form $y = y_\delta^*(x)$, or 
\begin{align*}
\widetilde{n}^*_\delta = \widetilde{n}^*_0 + \delta\,\widetilde{n}^*_1 + \delta^2\,\widetilde{n}^*_2 + \dots,\quad \Phi^*_\delta = \Phi^*_0 + \delta\,\Phi^*_1 + \delta^2\,\Phi^*_2 + \dots,
\end{align*}
where each of the coefficients $\widetilde{n}^*_k$, $\Phi^*_k$ is a uniquely-determined functional of $(n_0,\bm{\pi},\varrho)$. Perturbatively solving the invariance equation leads to the simple relations $\Phi^*_0 = \Phi^*_1 = 0$ and $\widetilde{n}^*_0 =  0$, as well as the first non-trivial coefficient
\begin{align}
\widetilde{n}^*_1 = \partial_{\bm{q}}\partial_{\bm{q}}:\bigg(n^{-1}\bm{P}\bm{P} + n\left\langle\bm{\xi}\bm{\xi}\right\rangle\bigg) - \frac{1}{\epsilon}\,\partial_{\bm{q}}\cdot\left(B\mathbb{J}\frac{\bm{\pi}}{n_0}\right).\label{qn_n1}
\end{align}
This expression for $\widetilde{n}^*_1$ specifies the leading-order deviation in the electron density away from the nominal value $n_0$ in a quasineutral motion of the Vlasov-Poisson system. Physically, the formula indicates that any compressive stress experienced by the electrons produces a compensating electrostatic field that maintains incompressibility of the electron flow when the dynamics is quasineutral.

The QNVP system defined in Def.\,\ref{def1} may now be recovered as the leading-order slow manifold reduction of the VP system. The details of this simple calculation are omitted. Notice the QNVP system resembles the incompressible Euler equations with a kinetic closure for the pressure. In contrast to the incompressible Euler equations, in which the scalar pressure adjusts to maintain incompressibility, quasineutral electrostatic plasmas maintain incompressibility by generating compensating electric fields. This physical difference between the two systems is perhaps most vividly illustrated the the nature of the fast oscillations that the two models omit: incompressible Euler omits sound waves (traveling waves), while magnetized quasineutral Vlasov-Poisson omits Langmuir oscillations (standing waves).

\section{Hamiltonian formulation\label{sec:ham_formulation}}
This Section applies the theory of Poisson-Dirac submanifolds\cite{crainic_integrability_2004} to deduce a Hamiltonian formulation of the QNVP system. The argument presented here first establishes the important technical result that the $\delta = 0$ slow manifold (described above) inherits a Poisson bracket from the ambient phase space $\mathfrak{s}^*$.
\begin{proposition}\label{PD_proposition}
Let $\mathfrak{s}^*$ denote the space of tuples $(n,\bm{P},\varrho)$, where $n:\mathbb{T}^2\rightarrow\mathbb{R}$, $\bm{P}:\mathbb{T}^2\rightarrow\mathbb{R}$, and $\varrho:\mathbb{T}^2\times\mathbb{R}^2\rightarrow\mathbb{R}$. The submanifold $\Sigma\subset\mathfrak{s}^*$ defined by 
\begin{align*}
\Sigma = \{(n,\bm{P},\varrho)\in\mathfrak{s}^*\mid \partial_{\bm{q}}n = 0,\quad \partial_{\bm{q}}\cdot\bm{P} = 0\}
\end{align*}
is a Poisson-Dirac submanifold when $\mathfrak{s}^*$ is equipped with the Poisson bracket $\{\cdot,\cdot\}_E$ from Eq.\,\eqref{VP_cbracket}. In particular, there is a Poisson bracket $\{\cdot,\cdot\}_\Sigma$ on $\Sigma$ induced by $\{\cdot,\cdot\}_E$ given explicitly in Eq.\,\eqref{QNVP_cbracket}.
\end{proposition}
\noindent The argument concludes by demonstrating that the natural Poisson bracket structure on the $\delta = 0$ slow manifold leads to the claimed Hamiltonian formulation for the QNVP system.
\begin{proposition}\label{hamiltonian_formulation_proposition}
The magnetized quasineutral Vlasov-Poisson system (cf Def. \ref{def1}) is a Hamiltonian system on $\Sigma$ with Hamiltonian $\mathcal{H}_\Sigma = \lim_{\delta\rightarrow 0}\mathcal{H}_{E}\mid \Sigma$ and Poisson bracket $\{\cdot,\cdot\}_\Sigma$. Here $\mathcal{H}_E$ is defined in Eq.\,\eqref{VP_chamiltonian} and $\{\cdot,\cdot\}_\Sigma$ is provided by Proposition \ref{PD_proposition}.
\end{proposition}

\begin{proof}[Proof of Proposition \ref{PD_proposition}]
To establish that $\Sigma$ is Poisson-Dirac, we must show (A) that it satisfies the Poisson-Dirac condition $T_\sigma \Sigma \cap \widehat{\pi}_E(T_\sigma\Sigma^\circ) = \{0\}$, and (B) that the induced bracket guaranteed by (A) varies smoothly along $\Sigma$. Here $\widehat{\pi}_E$ denotes the bundle map $T^*\mathfrak{s}^*\rightarrow T\mathfrak{s}^*$ associated with $\{\cdot,\cdot\}_E$ and $T_\sigma\Sigma^\circ\subset T^*_\sigma\mathfrak{s}^*$ denotes the annihilator of $T_\sigma\Sigma\subset T_\sigma\mathfrak{s}^*$. 

(A) Let $\sigma = (n_0,\bm{\pi},\varrho)\in\Sigma$ be a general point in $\Sigma$. Here, $n_0$ is a real constant and $\bm{\pi}$ is a divergence-free vector field on $\mathbb{T}^2$. The tangent space to $\Sigma$ at $\sigma$ is given by 
\begin{align*}
T_\sigma\Sigma = \{(\delta n_0,\delta\bm{\pi},\delta\varrho)\mid \delta n_0\in \mathbb{R},\quad \delta\bm{\pi}:\mathbb{T}^2\rightarrow\mathbb{R}^2,\quad \delta\varrho:\mathbb{T}^2\times \mathbb{R}^2\rightarrow\mathbb{R}, \quad \partial_{\bm{q}}\cdot\delta\bm{\pi} = 0\}.
\end{align*}
The annihilator of $T_\sigma\Sigma$ in $T^*_\sigma\mathfrak{s}^*\ni (\delta n^*,\delta\bm{P}^*,\delta\varrho^*)$ is therefore
\begin{align*}
T_\sigma\Sigma^\circ = \{(\delta n^*,\delta\bm{P}^*,\delta\varrho^*)\in T^*_\sigma\mathfrak{s}^*\mid \exists\delta\Phi^*:\mathbb{T}^2\rightarrow\mathbb{R},\quad\int\delta\Phi^*\,d\bm{q} = 0,\quad \int \delta n^*\,d\bm{q} = 0,\quad \delta\bm{P}^* =\partial_{\bm{q}}\delta\Phi^*,\quad \delta\varrho^* = 0\}.
\end{align*}
Suppose that the tangent vector $\delta z =(\delta n,\delta\bm{P},\delta\varrho) \in T_\sigma\mathfrak{s}^*$ is contained in both $T_\sigma\Sigma$ and $\widehat{\pi}_E(T_\sigma\Sigma^\circ)$. Since $\delta z\in \widehat{\pi}_E(T_\sigma\Sigma^\circ)$ there must be $\delta\alpha=(\delta n^*,\delta\bm{P}^*,\delta\varrho^*) = (\delta n^*,\partial_{\bm{q}}\delta\Phi^*,0)\in T_\sigma\Sigma^\circ$ such that $\delta z = \widehat{\pi}_E(\delta \alpha)$. Notice that $\widehat{\pi}_E(\delta \alpha)$ can be computed explicitly using Eqs.\,\eqref{ndot_G}-\eqref{varrhodot_G} by making the substitutions $\delta G/\delta n \rightarrow \delta n^*$, $\delta G/\delta\bm{P} \rightarrow \delta\bm{P}^*$, and $\delta G/\delta\varrho \rightarrow \delta\varrho^*$. By Eq.\,\eqref{ndot_G} this implies $\delta n = -\partial_{\bm{q}}\cdot(n_0\partial_{\bm{q}}\delta\Phi^*) = - n_0\,\Delta\delta\Phi^*$. On the other hand, $\delta z\in T_\sigma \Sigma$ implies that $\delta n = \delta n_0$ must be constant. Integrating $\delta n_0 = - n_0\,\Delta \delta\Phi^*$ over $\mathbb{T}^2$ implies that constant must vanish, $\delta n_0 = 0$. Thus, $\Delta \delta\Phi^* = 0$, which implies $\delta\Phi^* = 0$ because $\delta\Phi^*$ has zero mean. It follows that $\delta n = -n_0\Delta\delta\Phi^* = 0$. By Eq.\,\eqref{Pdot_G}, $\delta z = \widehat{\pi}_E(\delta\alpha)$ also implies $\delta\bm{P} = -\epsilon\,n_0\partial_{\bm{q}}\delta n^*$. But $\delta z\in T_\sigma\Sigma$ implies $\delta\bm{P} = \delta\bm{\pi}$ is divergence-free. Therefore $\Delta \delta n^* = 0$, which requires $\delta n^* = 0$ because $\delta n^*$ has zero mean. It follows that $\delta\bm{P} = - \epsilon\,n_0\,\partial_{\bm{q}}\delta n^* = 0$. Finally, in light of $\delta n^* = \delta\Phi^* = 0$ and $\delta z = \widehat{\pi}_E(\delta\alpha)$, Eq.\,\eqref{varrhodot_G} implies $\delta\varrho = 0$. We have shown $\delta z \in T_\sigma\Sigma\cap \widehat{\pi}_E(T_\sigma\Sigma^\circ)$ implies $\delta z = 0$, which establishes the claim $T_\sigma\Sigma\cap \widehat{\pi}_E(T_\sigma\Sigma^\circ) = \{0\}$ for all $\sigma\in\Sigma$.

(B) Part (A) of the proof guarantees that there is a bilinear bracket operation $\{\cdot,\cdot\}_\Sigma$ between functionals $F,G:\Sigma\rightarrow\mathbb{R}$. To complete the proof we must demonstrate that $\{F,G\}_\Sigma$ is a smooth function on $\Sigma$ when $F,G$ are both smooth. For this demonstration we will find and analyze an explicit formula for $\{F,G\}_\Sigma$. The value of the bracket $\{F,G\}_\Sigma$ at $\sigma=(n_0,\bm{\pi},\varrho)\in \Sigma$ is given by $\{F,G\}_\Sigma(\sigma) = \widetilde{(dF_\sigma)}\cdot \pi_{E}\cdot \widetilde{(dG_\sigma)}$, where $\pi_E$ denotes the Poisson tensor on $\mathfrak{s}^*$ associated with $\{\cdot,\cdot\}_E$. Here $dF_\sigma,dG_\sigma\in T^*_\sigma\Sigma$ denote the differentials of $F,G$ at the point $\sigma$ and $\widetilde{(dF_\sigma)},\widetilde{(dG_\sigma)}\in T^*_\sigma\mathfrak{s}^*$ are any covectors in $T^*_\sigma\mathfrak{s}^*$ such that
\begin{align}
&\widetilde{(dF_\sigma)}\mid T_\sigma\Sigma = dF_\sigma,\quad \widetilde{(dG_\sigma)}\mid T_\sigma\Sigma = dG_\sigma\label{extension_property}\\
&\widetilde{(dF_\sigma)}\mid \widehat{\pi}_E(T_\sigma\Sigma^\circ) = 0,\quad \widetilde{(dG_\sigma)}\mid \widehat{\pi}_E(T_\sigma\Sigma^\circ) = 0.\label{transverse_property}
\end{align}
When applied to $\delta \sigma = (\delta n_0,\delta\bm{\pi},\delta\varrho)\in T_\sigma\Sigma$ the value of $dG_\sigma$ is
\begin{align*}
dG_\sigma(\delta \sigma) = \int \frac{\delta G}{\delta n_0}\,\delta n_0\,d\bm{q} + \int \frac{\delta G}{\delta \bm{\pi}}\cdot\delta\bm{\pi}\,d\bm{q} + \int\frac{\delta G}{\delta\varrho}\,\delta\varrho\,d\bm{\xi}\,d\bm{q},
\end{align*}
where $\delta G/\delta n_0$ is required to be constant and $\delta G/\delta\bm{\pi}$ is required to be divergence-free.
Suppose that the value of $\widetilde{dG_\sigma}$ when applied to $\delta z = (\delta n,\delta \bm{P},\delta\varrho)\in T_\sigma\mathfrak{s}^*$ is given by
\begin{align*}
\widetilde{dG_\sigma}(\delta z) = \int \widetilde{\frac{\delta G}{\delta n_0}}\,\delta n\,d\bm{q} + \int \widetilde{\frac{\delta G}{\delta\bm{\pi}}}\cdot \delta\bm{P}\,d\bm{q} + \int \widetilde{\frac{\delta G}{\delta\varrho}}\,\delta\varrho\,d\bm{\xi}\,d\bm{q},
\end{align*}
where $\widetilde{\frac{\delta G}{\delta n_0}}$, $\widetilde{\frac{\delta G}{\delta\bm{\pi}}}$, and $\widetilde{\frac{\delta G}{\delta\varrho}}$ are unknown coefficients. The condition \eqref{extension_property} implies
\begin{align*}
\widetilde{\frac{\delta G}{\delta n_0}}  = \frac{\delta G}{\delta n_0} + \psi_G,\quad \widetilde{\frac{\delta G}{\delta\bm{\pi}}}  = \frac{\delta G}{\delta\bm{\pi}} + \partial_{\bm{q}}\Phi_G,\quad \widetilde{\frac{\delta G}{\delta\varrho}}  = \frac{\delta G}{\delta \varrho},
\end{align*}
where $\psi_G$ and $\Phi_G$ are both functions on $\mathbb{T}^2$ with zero mean, not determined by condition \eqref{extension_property}. In order for $\widetilde{dG_\sigma}$ to satisfy property \eqref{transverse_property}, for each $\alpha = (\delta n^*,\partial_{\bm{q}}\delta\Phi^*,0)\in T_\sigma\Sigma^0$ the tangent vector $\delta z = \widehat{\pi}_E(\alpha)\in T_\sigma\mathfrak{s}^*$ must annihilate $\widetilde{dG_\sigma}$. An explicit expression for $\delta z = (\delta n,\delta\bm{P},\delta\varrho)$ follows from Eqs.\,\eqref{ndot_G}-\eqref{varrhodot_G} using the substitutions $\delta G/\delta n\rightarrow \delta n^*$, $\delta G/\delta \bm{P}\rightarrow \partial_{\bm{q}}\delta \Phi^*$, and $\delta G/\delta\varrho\rightarrow 0$. The result is
\begin{align}
&\delta n +\epsilon\partial_{\bm{q}}\cdot\left(n_0\partial_{\bm{q}}\delta\Phi^*\right) = 0 \label{delta_n_biv}\\
&\delta\bm{P} + \epsilon\,\partial_{\bm{q}}\left( {\bm{\pi}}\cdot \partial_{\bm{q}}\delta\Phi^*\right) + \epsilon\,\partial_{\bm{q}}\cdot\left(\partial_{\bm{q}}\delta\Phi^*\right){\bm{\pi}} + \epsilon\,n_0\,\Omega\mathbb{J}\partial_{\bm{q}}\delta\Phi^* = - \epsilon\,n_0\,\partial_{\bm{q}}\delta n^* + n_0\,B\,\mathbb{J}\partial_{\bm{q}}\delta\Phi^* \label{delta_P_biv}\\
&\delta\varrho + \bigg[\epsilon\,\partial_{\bm{q}}\delta\Phi^* \bigg]\cdot\partial_{\bm{q}}\varrho- \partial_{\bm{\xi}}\cdot\bigg(\bigg[\epsilon\,\partial_{\bm{q}}\left(\bm{\xi}\cdot\partial_{\bm{q}}\delta\Phi^*\right)\bigg]\varrho\bigg) = 0.\label{delta_varrho_biv}
\end{align}
The condition $\widetilde{dG_\sigma}(\delta z) = 0$ therefore implies
\begin{align*}
&\Delta \Phi_G = 0\\
&\epsilon\,\Delta \psi_G = \partial_{\bm{q}}\cdot\bigg([B -\epsilon\,\Omega]\mathbb{J}\frac{\delta G}{\delta\bm{\pi}}\bigg) - \epsilon\,\Delta\left(\frac{\bm{\pi}}{n_0}\cdot \frac{\delta G}{\delta\bm{\pi}}\right) - \epsilon\,\partial_{\bm{q}}\partial_{\bm{q}}:\left\langle\left(\partial_{\bm{\xi}}\frac{1}{n_0}\frac{\delta G}{\delta \varrho}\right)\bm{\xi}\right\rangle + \epsilon\,\partial_{\bm{q}}\cdot\left\langle \frac{1}{n_0}\frac{\delta G}{\delta \varrho}\partial_{\bm{q}}\ln\varrho\right\rangle,
\end{align*}
which uniquely determines $\Phi_G$ and $\psi_G$. Note in particular that $\Phi_G = 0$. The bracket $\{F,G\}_\Sigma = \widetilde{dF_\sigma}\cdot\pi_E\cdot\widetilde{dG_\sigma}$ may now be computed explicitly by making the following substitutions in Eq.\,\eqref{VP_cbracket}:
\begin{align*}
\frac{\delta G}{\delta n}&\rightarrow \frac{\delta G}{\delta n_0} -\frac{\bm{\pi}}{n_0}\cdot\frac{\delta G}{\delta\bm{\pi}} + \mathcal{G}\bigg(\partial_{\bm{q}}\cdot\bigg[\frac{1}{\epsilon}(B-\epsilon\,\Omega)\mathbb{J}\frac{\delta G}{\delta\bm{\pi}} - \partial_{\bm{q}}\cdot\left(\left\langle\partial_{\bm{\xi}}\left(\frac{1}{n_0}\frac{\delta G}{\delta\varrho}\right)\bm{\xi}\right\rangle\right) + \partial_{\bm{q}}\left\langle\frac{1}{n_0}\frac{\delta G}{\delta \varrho}\right\rangle - \left\langle\partial_{\bm{q}}\frac{1}{n_0}\frac{\delta G}{\delta\varrho}\right\rangle\bigg]\bigg)+\text{const.}\\
\frac{\delta G}{\delta\bm{P}}&\rightarrow \frac{\delta G}{\delta\bm{\pi}}\\
\frac{\delta G}{\delta\varrho} & \rightarrow\frac{\delta G}{\delta\varrho}.
\end{align*}
The result is
\begin{align}
&\{F,G\}_{\Sigma} = \int (B-\epsilon\,\Omega)\,\frac{\delta F}{\delta\bm{\pi}}\cdot \mathbb{J}\frac{\delta G}{\delta\bm{\pi}} \,n_0\,d\bm{q} \nonumber\\
&-\epsilon\int\bigg(\frac{\delta F}{\delta\bm{\pi}}\cdot \bigg[    \left\langle\partial_{\bm{q}}\frac{\delta G}{\delta\varrho}\right\rangle\bigg]-\frac{\delta G}{\delta\bm{\pi}}\cdot \bigg[  \left\langle\partial_{\bm{q}}\frac{\delta F}{\delta\varrho}\right\rangle\bigg]\bigg)\,d\bm{q}\nonumber\\
&- \epsilon \int \bigg(\left[\frac{\delta F}{\delta\bm{\pi}} - \left\langle\partial_{\bm{\xi}}\frac{1}{n_0}\frac{\delta F}{\delta\varrho}\right\rangle\right]\cdot\partial_{\bm{q}}\cdot\left[\left\langle\partial_{\bm{\xi}}\left(\frac{\delta G}{\delta\varrho}\right)\bm{\xi}\right\rangle\right] -\left[\frac{\delta G}{\delta\bm{\pi}} - \left\langle\partial_{\bm{\xi}}\frac{1}{n_0}\frac{\delta G}{\delta\varrho}\right\rangle\right]\cdot\partial_{\bm{q}}\cdot\left[\left\langle\partial_{\bm{\xi}}\left(\frac{\delta F}{\delta\varrho}\right)\bm{\xi}\right\rangle\right] \bigg)\,d\bm{q}\nonumber\\
&+\epsilon\int \bigg(\partial_{\bm{q}}\frac{1}{n_0}\frac{\delta F}{\delta \varrho}\cdot \partial_{\bm{\xi}}\frac{1}{n_0}\frac{\delta G}{\delta\varrho} -\partial_{\bm{q}}\frac{1}{n_0}\frac{\delta G}{\delta \varrho}\cdot \partial_{\bm{\xi}}\frac{1}{n_0}\frac{\delta F}{\delta\varrho} \bigg)\,n_0\,\varrho\,d\bm{\xi}\,d\bm{q} + \int (B-\epsilon\,\Omega)\bigg(\partial_{\bm{\xi}}\frac{1}{n_0}\frac{\delta F}{\delta\varrho}\bigg)\cdot\mathbb{J}\bigg(\partial_{\bm{\xi}}\frac{1}{n_0}\frac{\delta G}{\delta\varrho}\bigg)\,n_0\,\varrho\,d\bm{\xi}\,d\bm{q}\nonumber\\
&-\epsilon\int \bigg(\left\langle\partial_{\bm{q}}\frac{1}{n_0}\frac{\delta F}{\delta\varrho}\right\rangle\cdot\left\langle\partial_{\bm{\xi}}\frac{1}{n_0}\frac{\delta G}{\delta\varrho}\right\rangle-\left\langle\partial_{\bm{q}}\frac{1}{n_0}\frac{\delta G}{\delta\varrho}\right\rangle\cdot\left\langle\partial_{\bm{\xi}}\frac{1}{n_0}\frac{\delta F}{\delta\varrho}\right\rangle\bigg)\,n\,d\bm{q}-\int (B-\epsilon\,\Omega)\, \left\langle\partial_{\bm{\xi}}\frac{1}{n_0}\frac{\delta F}{\delta\varrho}\right\rangle\cdot \mathbb{J} \left\langle\partial_{\bm{\xi}}\frac{1}{n_0}\frac{\delta G}{\delta\varrho}\right\rangle\,n_0\,d\bm{q}.\label{QNVP_cbracket}
\end{align}
Notice that the inverse Laplacian $\mathcal{G}$ does not appear in the bracket because of $\partial_{\bm{q}}\cdot (\delta G/\delta\bm{\pi}) = 0$. This bracket is polynomial in $(\bm{\pi},\varrho)$ and regular in $n_0$ away from $n_0 = 0$. It follows that $\Sigma$ is a Poisson-Dirac submanifold in $\mathfrak{s}^*$. 

\end{proof}

\begin{proof}[Proof of Proposition \ref{hamiltonian_formulation_proposition}]
On $\Sigma$ the fluctuating part of the density satisfies $n - n_0 = O(\delta^2)$, which implies that the energy stored in the electric field vanishes in the limit $\delta\rightarrow 0$. The Hamiltonian $\mathcal{H}_\Sigma$ is therefore the kinetic energy
\begin{align*}
\mathcal{H}_\Sigma(n_0,\bm{\pi},\varrho) & = \frac{1}{2}\int |\bm{\xi}|^2\,n_0\,\varrho\,d\bm{\xi}\,d\bm{q} + \frac{1}{2}\int|\bm{\pi}/n_0|^2\,n_0\,d\bm{q}.
\end{align*}
The functional derivatives of $\mathcal{H}_\Sigma$ are
\begin{gather*}
\frac{\delta\mathcal{H}_\Sigma}{\delta n_0} = \frac{1}{2}\int |\bm{\xi}|^2\,\varrho\,d\bm{\xi}\,d\bm{q} -\frac{1}{2}\int \left|\frac{\bm{\pi}}{n_0}\right|^2\,d\bm{q},\quad \frac{\delta\mathcal{H}_\Sigma}{\delta\bm{\pi}} = \frac{\bm{\pi}}{n_0},\quad \frac{\delta\mathcal{H}_\Sigma}{\delta \varrho} = \frac{1}{2}n_0\,|\bm{\xi}|^2.
\end{gather*}
It is now straightforward to compute the components of the Hamiltonian vector field $X_{\mathcal{H}_\Sigma} = (\dot{n}_0,\dot{\bm{\pi}},\dot{\varrho})$ using the formula \eqref{QNVP_cbracket} for $\{\cdot,\cdot\}_\Sigma$. 

To find $\dot{n}_0$ set $G = \mathcal{H}_\Sigma$ and $F = \int \psi_0\,n_0\,d\bm{q}$, where $\psi_0$ is a real constant. The bracket $\{F,G\}_{\Sigma} = \int \dot{n}_0\,\psi_0\,d\bm{q}$ vanishes. It follows that $\dot{n}_0 = 0$. This reproduces Eq.\,\eqref{qn_cont}.

To find $\dot{\bm{\pi}}$ set $G = \mathcal{H}_\Sigma$ and $F = \int \bm{w}\cdot\bm{\pi}\,d\bm{q}$, where $\bm{w}$ is any divergence-free vector field on $\mathbb{T}^2$. Computing the bracket $\{F,G\}_\Sigma = \int \bm{w}\cdot\dot{\bm{\pi}}\,d\bm{q}$  implies
\begin{align*}
\dot{\bm{\pi}} & = \Pi\bigg((B-\epsilon\,\Omega)\mathbb{J}{\bm{\pi}} - \epsilon\,n_0\,\partial_{\bm{q}}\cdot\langle\bm{\xi}\bm{\xi}\rangle\bigg).
\end{align*}
This recovers Eq.\,\eqref{qn_vorticity} in light of the vector identity $\partial_{\bm{q}}\cdot(n_0^{-1}\bm{\pi}\bm{\pi}) = n_0^{-1}\partial_{\bm{q}}|\bm{\pi}|^2 + \Omega\mathbb{J}\bm{\pi}.$

To find $\dot{\varrho}$ set $G = \mathcal{H}_\Sigma$ and $F = \int \chi\,\varrho\,d\bm{\xi}\,d\bm{q}$, where $\chi$ is an arbitrary function on phase space space. Computing the bracket $\{F,G\}_\Sigma = \int \chi\dot{\varrho}\,d\bm{\xi}\,d\bm{q}$ implies
\begin{align*}
\dot{\varrho}&+ \partial_{\bm{q}}\cdot\bigg(\bigg[\epsilon\frac{\bm{\pi}}{n_0} + \epsilon\,\bm{\xi} \bigg]\varrho\bigg) + \partial_{\bm{\xi}}\cdot\bigg(\bigg[\epsilon\,\partial_{\bm{q}}\cdot\left\langle\bm{\xi}\bm{\xi}\right\rangle - \epsilon\,\partial_{\bm{q}}\left(\frac{\bm{\pi}}{n_0}\right)\cdot\bm{\xi} + (B-\epsilon\,\Omega)\mathbb{J}\bm{\xi}\bigg]\varrho\bigg)\\
& - \partial_{\bm{q}}\cdot\bigg(\bigg[ \epsilon\,\left\langle \bm{\xi}\right\rangle\bigg]\varrho\bigg) + \partial_{\bm{\xi}}\cdot\bigg(\bigg[\epsilon\,\partial_{\bm{q}}\left\langle\bm{\xi}\right\rangle \cdot\bm{\xi} - (B-\epsilon\,\Omega)\mathbb{J}\left\langle\bm{\xi}\right\rangle\bigg]\varrho\bigg) = 0.
\end{align*}
This agrees with Eq.\,\eqref{qn_kinetic} in light of the vector identity $(\partial_{\bm{q}}\bm{\pi})\cdot\bm{\xi} = \bm{\xi}\cdot\partial_{\bm{q}}\bm{\pi} - \Omega\,\mathbb{J}\bm{\xi}$. This completes the demonstration that Hamilton's equations with Hamiltonian $\mathcal{H}_\Sigma$ and Poisson bracket $\{\cdot,\cdot\}_\Sigma$ agree with the magnetized quasineutral Vlasov-Poisson system, as claimed.
\end{proof}

\section{Discussion}
The above analysis used slow manifold reduction and elements of Poisson geometry to find a Hamiltonian formulation for the QNVP system. This calculation fits into a more general pattern\cite{r_s_mackay_slow_2004,j_w_burby_slow_2020,burby_slow_2022}, wherein non-dissipative reduced models formulated as slow manifold reductions naturally inherit Hamiltonian formulations from their parent models. In the context of models for plasmas, the earliest demonstration of a slow manifold inheriting Hamiltonian structure was given in Ref.\,\onlinecite{burby_magnetohydrodynamic_2017}, where symplectic methods were used to deduce the Hamiltonian formulation of ideal magnetohydrodynamics, as well as its higher-order corrections. Similar methods were used to identify Hamiltonian formulations of the guiding center plasma model\cite{burby_hamiltonian_2018}, perturbative corrections of the Vlasov-Poisson model\cite{miloshevich_hamiltonian_2021}, and nonlinear WKB reductions of variational fluid models\cite{burby_variational_2020}. The symplectic approach to inheritance in fluid and kinetic plasma models usually entails introducing Lagrangian variables in order to obtain a symplectic formulation of the problem, reducing to a symplectic submanifold, and then quotienting by particle relabeling symmetry in order to find an Eulerian expression of the Hamiltonian structure. The Poisson geometric approach to inheritance used in this Article, which was dubbed the Poisson-Dirac constraint method in \cite{Pinto_Burby_2025}, is more efficient than the earlier symplectic methods because it does not require introducing Lagrangian variables. On the other hand, it is currently unclear how to handle higher-order corrections to the slow manifold in the Poisson setting. While the general theory of Poisson-Dirac submanifolds leads to brackets on perturbed slow manifolds as perturbation series, it is unclear how to truncate those series while maintaining the Jacobi identity. This stands in contrast to the symplectic case, where near-identity transformations of the perturbed slow manifold can be introduced that cause the perturbation series for the symplectic form to truncate at finite order. This dichotomy is caused by the relative complexity of deformation theory for Poisson structures when compared with deformation theory for symplectic structures.

The slow manifold analysis presented here used $\delta$, the ratio of Debye length to field scale length, as the timescale separation parameter. More precisely, the limit studied was $\delta\rightarrow 0$ holding $\epsilon,\Lambda$ fixed. In future studies of higher-order corrections to the QNVP system in the strongly magnetized regime it would be better to use $\overline{\delta} = \delta/\epsilon$ in place of $\delta$. Doing so will alleviate the issue suggested by Eq.\,\eqref{qn_n1}, where the quasineutral slow manifold diverges as $\epsilon\rightarrow 0$ at fixed $\delta$. 

The Hamiltonian formulation of the QNVP model presented here may admit a concomittant variational formulation. As presented here, the model takes the form of a hybrid fluid-kinetic system. Tronci\cite{tronci_lagrangian_2013} previously developed methods for identifying variational formulations of this type of model, which have been subsequently extended to hybrid quantum-classical model, e.g. Ref.\,\onlinecite{gay-balmaz_evolution_2022}. It may be interesting to extend Tronci's ideas to the QNVP system.

This Article studied the compound asymptotic limit where quasineutrality is applied after applying the electrostatic approximation. Since both quasineutrality and electrostatics correspond to singular limits of the Vlasov-Maxwell model, it is possible that a different model emerges when applying the electrostatic approximation \emph{after} the quasineutral approximation. The refactoring of the quasineutral-electrostatic compound asymptotic limit does not appear to have been discussed by previous authors. It would be interesting to both formulate the model and find its Hamiltonian structure.

The formal analysis presented here assumes a doubly-periodic planar spatial domain. There do not appear to be any essential obstacles to extending our calculations to a three-dimensional triply-periodic spatial domain. However, including a spatial boundary, in any space dimension, would introduce genuine theoretical challenges. In fact it is totally unknown whether Vlasov-Poisson in a bounded spatial domain comprises an infinite-dimensional Hamiltonian system.

Solutions of the VP model initialized on the quasineutral slow manifold will exhibit only very small-amplitude oscillations at the Lagmuir oscillation timescale. However, solutions initialized close to but not exactly on the slow manifold will excite dynamically-important Langmuir oscillations, described to leading order by Eq.\,\eqref{nearly-periodic}. Slow manifold reduction does not capture the effects of these oscillations. On the other hand, Eq.\,\eqref{nearly-periodic} also demonstrates that the VP system, when expressed in terms of $(n_0,\widetilde{n},\bm{\Phi},\bm{\pi},\varrho)$ as in Eqs.\,\eqref{fs_mean_density}-\eqref{fs_vorticity}, comprises a nearly-periodic system in the sense of Kruskal\cite{kruskal_asymptotic_1962,burby_general_2020,burby_normal_2021,burby_nearly_2023}. Consequently, the VP system admits an formal $U(1)$-symmetry to all orders in $\delta$. It would be interesting to demonstrate this formal $U(1)$-symmetry admits a formal momentum map. If true, this momentum map would provide an adiabatic invariant for the VP system with the physical interpretation of Langmuir oscillation wave action\cite{brizard_wave-action_1993}. While it is already clear how to deduce the adiabatic invariant associated with Kruskal's $U(1)$ symmetry on symplectic and pre-symplectic phase space\cite{burby_general_2020}, the necessary theory has yet to be developed for Poisson phase spaces. In light of the recent formulation of a non-perturbative guiding center model\cite{j_w_burby_nonperturbative_2025}, this line of research offers a possible avenue to formulating a non-perturbative extension of the QNVP model, allowing for larger Debye lengths.

\begin{acknowledgments}
This material is based on work supported by the U.S. Department of Energy, Office of Science, Office of Advanced Scientific Computing Research, as a part of the Mathematical Multifaceted Integrated Capability Centers program, under Award Number DE-SC0023164. It was also supported by U.S. Department of Energy grant \# DE-FG02-04ER54742. The contributions of DAK and GNT in this work were conducted in the framework of participation of the University of Ioannina in the National Programme for the Controlled Thermonuclear Fusion, Hellenic Republic. E. Tassi acknowledges support from the GNFM.
\end{acknowledgments}

\section*{Data Availability Statement}
Data sharing not applicable – no new data generated

\nocite{}
\bibliography{references,add_refs}

\end{document}